\documentclass[showpacs,tightenlines,12pt,onecolumn]{revtex4}%
\usepackage{amsfonts}
\usepackage{amsmath}
\usepackage{amssymb}
\usepackage{graphicx}%
\newtheorem{theorem}{Theorem}

\newenvironment{proof}[1][Proof]{\noindent\textbf{#1.} }{\ \rule{0.5em}{0.5em}}
\begin{document}
\title{Two kinds of quantum adiabatic approximation}
\author{Ming-Yong Ye}
\email{myye@mail.ustc.edu.cn}
\author{Xiang-Fa Zhou}
\author{Yong-Sheng Zhang}
\author{Guang-Can Guo}
\affiliation{Key Laboratory of Quantum Information, Department of Physics, University of
Science and Technology of China, Hefei 230026, People's Republic of China}

\begin{abstract}
A simple proof of quantum adiabatic theorem is provided. Quantum adiabatic
approximation is divided into two kinds. For Hamiltonian $H\left(  t/T\right)
$, a relation between the size of the error caused by quantum adiabatic
approximation and the parameter $T$ is given.

\end{abstract}

\maketitle

\section{Introduction}

A quantum system is described by its Hamiltonian. When the Hamiltonian is
time-independent and its spectral decomposition is known, we can easily get
the dynamical evolution operator (DEO). However it is usually impossible to
get an analytic expression of the DEO when the Hamiltonian varies with time.
In a pioneering paper, Born and Fock considered a kind of time-dependent
Hamiltonian $H_{t}$ that has a special form $H_{t}=H\left(  s\right)  $ where
$s=t/T$ and the spectrum of $H\left(  s\right)  $ consists of purely discrete
eigenvalues \cite{born}. Their result and the extended results
\cite{kato,nenciu,messiah,avron} on more general $H\left(  s\right)  $ are all
called quantum adiabatic theorem (QAT). For a history of QAT we refer to Ref.
\cite{avron}.

Recently a debate arised on the validity of application of quantum adiabatic
theorem (QAT) or quantum adiabatic approximation (QAA)
\cite{ms,duki,ma,ms2,tong,sarandy,wu}. One thing can be sure is that the
widely used simple condition for QAA is actually insufficient \cite{ms,tong}.
Since QAA and the related Berry phase \cite{berry} have a wide application in
many fields \cite{bookgeometric}, it is valuable to find new conditions for
the approximation. The recent discussion on QAT or QAA is also stimulated by
quantum adiabatic computation \cite{jisuan}; quantum computers are believed to
be more powerful than computers that we use today \cite{nielsen}.

To eliminate ambiguity, the thing should be specified first is the rigorous
definitions of QAT and QAA. We only concern time-dependent matrix Hamiltonians
without energy degeneracy, which are also the topics of the recent debate. QAT
can only be discussed for Hamiltonians that have special form $H_{t}=H\left(
s\right)  $ where $s=t/T$. This can be seen from the original papers on QAT
\cite{born,messiah}. We think QAT states that the difference between the DEO
and the adiabatic evolution operator (AEO) (defined below) of the system will
approach zero in the limit $T\rightarrow\infty$. However, we can talk about
QAA for Hamiltonians that may not be written in the form $H_{t}=H\left(
s\right)  $. We define QAA as an approximation that uses the AEO in place of
the DEO in our calculations. Under this definition of QAA, we find that there
are two kinds of QAA and conditions for them to be acceptable are different.
Of course we can also discuss QAA for Hamiltonians that are written in the
form $H_{t}=H\left(  s\right)  $, one can discuss the relation between the
amplitude of the error caused by QAA and the value of the parameter $T$
\cite{regev}.

The structure of the paper is as follows. In Sec. \ref{sec:2}, we give an
introduction to QAT and present a proof of the theorem. In Sec. \ref{sec:3},
we give a discussion on QAA. In Sec. \ref{sec:4}, we give an example to
demonstrate the results obtained in Sec. \ref{sec:3}. In Sec. \ref{sec:5}, we
give a relation between the size of the error caused by QAA and the parameter
$T$ for Hamiltonian $H\left(  t/T\right)  $. Finally the conclusion is given
in Sec. \ref{sec:6}.

\section{Quantum adiabatic theorem\label{sec:2}}

We just consider non-degenerate two-level systems to demonstrate our basic
idea through out the paper. A smooth time-dependent Hamiltonian $H_{t}$ has
two instantaneous eigenvalues and eigenstates,%
\begin{equation}
H_{t}\left\vert m_{t}\right\rangle =m_{t}\left\vert m_{t}\right\rangle ,\text{
}m=1,2. \label{z1}%
\end{equation}
Eq. (\ref{z1}) cannot determine the phases of the eigenstates $\left\vert
m_{t}\right\rangle $. We choose the phases such that \cite{schiff}
\begin{equation}
\left\langle m_{t}\right\vert \frac{d}{dt}\left\vert m_{t}\right\rangle
=0,\text{ }m=1,2. \label{r0}%
\end{equation}
The DEO is denoted by $U_{d}\left(  t\right)  $ which satisfies the
Schr\"{o}dinger equation
\begin{equation}
i\hbar\frac{dU_{d}\left(  t\right)  }{dt}=H_{t}U_{d}\left(  t\right)  ,\text{
}U_{d}\left(  0\right)  =I. \label{z2}%
\end{equation}
Except for the trivial cases where the projectors $\left\vert m_{t}%
\right\rangle \left\langle m_{t}\right\vert $ are independent of time $t$, it
is usually impossible to obtain an analytic expression for $U_{d}\left(
t\right)  $. The AEO is defined as%
\begin{equation}
U_{a}\left(  t\right)  =%
{\displaystyle\sum\limits_{n=1}^{2}}
e^{-\frac{i}{\hbar}\int_{0}^{t}n_{t^{\prime}}dt^{\prime}}\left\vert
n_{t}\right\rangle \left\langle n_{0}\right\vert , \label{aeo}%
\end{equation}
which is a little different from the so-called adiabatic transformation
\cite{nenciu}.

\begin{theorem}
(QAT) When $H_{t}$ has the special form $H_{t}=H\left(  s\right)  $ where
$s=t/T\in\left[  0,1\right]  $, i.e., it varies from $H\left(  0\right)  $ to
$H\left(  1\right)  $ using time $T$, for any initial state $\left\vert
\Psi\left(  0\right)  \right\rangle $ there is
\begin{equation}
\lim_{T\rightarrow\infty}\left\Vert U_{d}\left(  t\right)  \left\vert
\Psi\left(  0\right)  \right\rangle -U_{a}\left(  t\right)  \left\vert
\Psi\left(  0\right)  \right\rangle \right\Vert =0. \label{main}%
\end{equation}

\end{theorem}

\textbf{Remark}. What we consider is the simplest kind of QAT and it has been
proved long ago by considering operator evolutions
\cite{born,kato,messiah,nenciu}, and no one has doubt on it. From the theorem
it is not hard to see that when the system starts from the ground state of
$H\left(  0\right)  $ it will evolve closely to the ground state of $H\left(
1\right)  $ if $T$ is big enough; this is the basis of quantum adiabatic
computation \cite{jisuan}. Now we give a proof by considering the state vector
evolution, which we think is more intuitive, and the process of the proof will
be used when we discuss QAA.

\begin{proof}
We write the system state at time $t$ as
\begin{equation}
U_{d}\left(  t\right)  \left\vert \Psi\left(  0\right)  \right\rangle =%
{\displaystyle\sum\limits_{n=1}^{2}}
c_{n}\left(  t\right)  e^{-\frac{i}{\hbar}\int_{0}^{t}n_{t^{\prime}}%
dt^{\prime}}\left\vert n_{t}\right\rangle . \label{nn}%
\end{equation}
Since $U_{d}\left(  t\right)  \left\vert \Psi\left(  0\right)  \right\rangle $
satisfies the Schr\"{o}dinger equation, we can obtain \cite{schiff}%
\begin{equation}
\frac{dc_{1}\left(  t\right)  }{dt}=-c_{2}\left(  t\right)  \left\langle
1_{t}\right\vert \frac{d}{dt}\left\vert 2_{t}\right\rangle e^{\frac{i}{\hbar
}\int_{0}^{t}\left(  1_{t^{\prime}}-2_{t^{\prime}}\right)  dt^{\prime}},
\label{ec1}%
\end{equation}%
\begin{equation}
\frac{dc_{2}\left(  t\right)  }{dt}=-c_{1}\left(  t\right)  \left\langle
2_{t}\right\vert \frac{d}{dt}\left\vert 1_{t}\right\rangle e^{\frac{i}{\hbar
}\int_{0}^{t}\left(  2_{t^{\prime}}-1_{t^{\prime}}\right)  dt^{\prime}}.
\label{ec2}%
\end{equation}
Integrate both sides of Eq. (\ref{ec1}) we get%
\begin{align}
&  c_{1}\left(  t^{\prime\prime}\right)  -c_{1}\left(  0\right) \nonumber\\
&  =-\int_{0}^{t^{\prime\prime}}c_{2}\left(  t\right)  \left\langle
1_{t}\right\vert \frac{d}{dt}\left\vert 2_{t}\right\rangle e^{\frac{i}{\hbar
}\int_{0}^{t}\left(  1_{t^{\prime}}-2_{t^{\prime}}\right)  dt^{\prime}%
}dt\nonumber\\
&  =-\int_{0}^{t^{\prime\prime}}c_{2}\left(  t\right)  \frac{\left\langle
1_{t}\right\vert \frac{d}{dt}\left\vert 2_{t}\right\rangle }{i\left(
1_{t}-2_{t}\right)  /\hbar}d\left[  e^{\frac{i}{\hbar}\int_{0}^{t}\left(
1_{t^{\prime}}-2_{t^{\prime}}\right)  dt^{\prime}}\right]  .
\end{align}
In the above we have used the fact that there is no energy degeneracy, i.e.,
$1_{t}-2_{t}\neq0$. Integrate by part we get
\begin{equation}
c_{1}\left(  t^{\prime\prime}\right)  -c_{1}\left(  0\right)  =A\left(
t^{\prime\prime}\right)  +B\left(  t^{\prime\prime}\right)  +C\left(
t^{\prime\prime}\right)  , \label{ABC}%
\end{equation}
where%
\begin{equation}
A\left(  t^{\prime\prime}\right)  =-c_{2}\left(  t\right)  f_{t}e^{\frac
{i}{\hbar}\int_{0}^{t}\left(  1_{t^{\prime}}-2_{t^{\prime}}\right)
dt^{\prime}}\left\vert _{t=0}^{t=t^{\prime\prime}}\right.  , \label{w1}%
\end{equation}%
\begin{equation}
B\left(  t^{\prime\prime}\right)  =\int_{0}^{t^{\prime\prime}}c_{2}\left(
t\right)  \frac{df_{t}}{dt}e^{\frac{i}{\hbar}\int_{0}^{t}\left(  1_{t^{\prime
}}-2_{t^{\prime}}\right)  dt^{\prime}}dt, \label{w2}%
\end{equation}%
\begin{equation}
C\left(  t^{\prime\prime}\right)  =\int_{0}^{t^{\prime\prime}}\left[
\frac{dc_{2}\left(  t\right)  }{dt}f_{t}e^{\frac{i}{\hbar}\int_{0}^{t}\left(
1_{t^{\prime}}-2_{t^{\prime}}\right)  dt^{\prime}}\right]  dt, \label{w3}%
\end{equation}%
\begin{equation}
f_{t}=\frac{\left\langle 1_{t}\right\vert \frac{d}{dt}\left\vert
2_{t}\right\rangle }{i\left(  1_{t}-2_{t}\right)  /\hbar}.
\end{equation}
Because $H\left(  s\right)  $ depends only on the parameter $s$, instantaneous
eigenstates and eigenvalues are also dependent only on $s$, i.e., we can write
$\left\vert m_{t}\right\rangle =\left\vert m\left(  s\right)  \right\rangle $,
$m_{t}=m\left(  s\right)  $. Now we have%
\begin{equation}
f_{t}=\frac{1}{T}\frac{\hbar\left\langle 1\left(  s\right)  \right\vert
\frac{d}{ds}\left\vert 2\left(  s\right)  \right\rangle }{i\left(  1\left(
s\right)  -2\left(  s\right)  \right)  }=\frac{f\left(  s\right)  }{T}.
\label{f}%
\end{equation}
From Eqs. (\ref{w1}-\ref{w3}) we can get%
\begin{equation}
\left\vert A\left(  t^{\prime\prime}\right)  \right\vert \leq\frac{2}{T}%
\max_{s\in\left[  0,1\right]  }\left\vert f\left(  s\right)  \right\vert ,
\label{a}%
\end{equation}%
\begin{equation}
\left\vert B\left(  t^{\prime\prime}\right)  \right\vert \leq\frac{1}{T}%
\max_{s\in\left[  0,1\right]  }\left\vert \frac{d}{ds}f\left(  s\right)
\right\vert , \label{bb}%
\end{equation}%
\begin{equation}
\left\vert C\left(  t^{\prime\prime}\right)  \right\vert \leq\frac{1}{T}%
\max_{s\in\left[  0,1\right]  }\left\vert \left\langle 2\left(  s\right)
\right\vert \frac{d}{ds}\left\vert 1\left(  s\right)  \right\rangle f\left(
s\right)  \right\vert . \label{c}%
\end{equation}
In deriving Eq. (\ref{c}) from Eq. (\ref{w3}) we have used Eq. (\ref{ec2}).
From Eq. (\ref{ABC}), we know $\left\vert A\left(  t^{\prime\prime}\right)
\right\vert +\left\vert B\left(  t^{\prime\prime}\right)  \right\vert
+\left\vert C\left(  t^{\prime\prime}\right)  \right\vert $ is an upper bound
of $\left\vert c_{1}\left(  t^{\prime\prime}\right)  -c_{1}\left(  0\right)
\right\vert $, and Eqs. (\ref{a}-\ref{c}) indicate this upper bound will
approach zero in the limit $T\rightarrow\infty$, so we have%
\begin{equation}
\lim_{T\rightarrow\infty}\left\vert c_{1}\left(  t\right)  -c_{1}\left(
0\right)  \right\vert =0. \label{mm1}%
\end{equation}
Similarly we can obtain%
\begin{equation}
\lim_{T\rightarrow\infty}\left\vert c_{2}\left(  t\right)  -c_{2}\left(
0\right)  \right\vert =0. \label{mm2}%
\end{equation}
Because
\begin{align}
&  \left\Vert U_{d}\left(  t\right)  \left\vert \Psi\left(  0\right)
\right\rangle -U_{a}\left(  t\right)  \left\vert \Psi\left(  0\right)
\right\rangle \right\Vert \nonumber\\
&  =\sqrt{%
{\displaystyle\sum\limits_{n=1}^{2}}
\left\vert c_{n}\left(  t\right)  -c_{n}\left(  0\right)  \right\vert ^{2}},
\label{mm}%
\end{align}
we get%
\begin{equation}
\lim_{T\rightarrow\infty}\left\Vert U_{d}\left(  t\right)  \left\vert
\Psi\left(  0\right)  \right\rangle -U_{a}\left(  t\right)  \left\vert
\Psi\left(  0\right)  \right\rangle \right\Vert =0,
\end{equation}
which completes the proof.
\end{proof}

\section{Quantum adiabatic approximation\label{sec:3}}

The system state $\left\vert \Psi\left(  t\right)  \right\rangle $ at time $t$
and the initial state $\left\vert \Psi\left(  0\right)  \right\rangle $ are
connected by the relation $\left\vert \Psi\left(  t\right)  \right\rangle
=U_{d}\left(  t\right)  \left\vert \Psi\left(  0\right)  \right\rangle $.
Suppose the operator $B\left(  t\right)  $ represents a physical quantity at
time $t$. Quantum mechanics tells us that when we measure the quantity the
result will be random and the average value will be
\begin{equation}
\left\langle \Psi\left(  t\right)  \right\vert B\left(  t\right)  \left\vert
\Psi\left(  t\right)  \right\rangle =\left\langle \Psi\left(  0\right)
\right\vert U_{d}^{\dagger}\left(  t\right)  B\left(  t\right)  U_{d}\left(
t\right)  \left\vert \Psi\left(  0\right)  \right\rangle . \label{add01}%
\end{equation}
Assume we are given the initial state $\left\vert \Psi\left(  0\right)
\right\rangle $ and the Hamiltonian $H_{t}$ of the system. If we cannot figure
out $U_{d}\left(  t\right)  $ from $H_{t}$, usually we will not know the
average value $\left\langle \Psi\left(  t\right)  \right\vert B\left(
t\right)  \left\vert \Psi\left(  t\right)  \right\rangle $. From QAT we know
that in some cases the difference between $U_{d}\left(  t\right)  $ and
$U_{a}\left(  t\right)  $ will be very small, we may consider whether it is
acceptable to use $U_{a}\left(  t\right)  $ in place of $U_{d}\left(
t\right)  $ in calculating the average in Eq. (\ref{add01}). We define QAA as
an approximation that uses the AEO $U_{a}\left(  t\right)  $ in place of the
DEO $U_{d}\left(  t\right)  $ in our calculations. The definition of QAA
applies to a general time-dependent $H_{t}$ contrast to QAT. If we require the
approximation to be acceptable for any physical quantity $B\left(  t\right)  $
and any initial state $\left\vert \Psi\left(  0\right)  \right\rangle $ in Eq.
(\ref{add01}), it means the difference between $U_{d}\left(  t\right)
\left\vert \Psi\left(  0\right)  \right\rangle $ and $U_{a}\left(  t\right)
\left\vert \Psi\left(  0\right)  \right\rangle $ should be small.

The difference between $U_{d}\left(  t\right)  \left\vert \Psi\left(
0\right)  \right\rangle $ and $U_{a}\left(  t\right)  \left\vert \Psi\left(
0\right)  \right\rangle $ is small means
\begin{equation}
\left\Vert U_{d}\left(  t\right)  \left\vert \Psi\left(  0\right)
\right\rangle -U_{a}\left(  t\right)  \left\vert \Psi\left(  0\right)
\right\rangle \right\Vert \ll1, \label{a1}%
\end{equation}
which is equivalent to
\begin{equation}
\left\langle \Psi\left(  0\right)  \right\vert U_{a}^{\dag}\left(  t\right)
U_{d}\left(  t\right)  \left\vert \Psi\left(  0\right)  \right\rangle
\approx1. \label{a2}%
\end{equation}
Due to the linearity of quantum mechanics, for arbitrary initial state
$\left\vert \Psi\left(  0\right)  \right\rangle $, condition (\ref{a2}) will
be satisfied when
\begin{equation}
\left\langle m_{0}\right\vert U_{a}^{\dag}\left(  t\right)  U_{d}\left(
t\right)  \left\vert m_{0}\right\rangle \approx1,\text{ }m=1,2, \label{a3}%
\end{equation}
which we regard as the condition for the first kind of QAA. This kind of QAA
pays attention to the phase of $U_{d}\left(  t\right)  \left\vert
m_{0}\right\rangle $. Condition (\ref{a3}) ensures that the relative phase
between the two instantaneous eigenstates in $U_{d}\left(  t\right)
\left\vert \Psi\left(  0\right)  \right\rangle $ is almost the same as that in
$U_{a}\left(  t\right)  \left\vert \Psi\left(  0\right)  \right\rangle $.

The probability of finding the system in an instantaneous energy eigenstate
$\left\vert n_{t}\right\rangle $ is of considerable interest, e.g., in
coherent population transfer among quantum states of atoms and molecules
\cite{rmp}. In this case, in Eq. (\ref{add01}) the operator $B\left(
t\right)  =\left\vert n_{t}\right\rangle \left\langle n_{t}\right\vert $ and
the probability is $\left\vert \left\langle n_{t}\right\vert U_{d}\left(
t\right)  \left\vert \Psi\left(  0\right)  \right\rangle \right\vert ^{2}$.
From the definition of the $U_{a}\left(  t\right)  $ in Eq. (\ref{aeo}) we
know $\left\vert \left\langle n_{t}\right\vert U_{a}\left(  t\right)
\left\vert \Psi\left(  0\right)  \right\rangle \right\vert ^{2}$ is a
constant, it means the system will follow the instantaneous eigenstate if it
starts from an instantaneous eigenstates. Therefore it is acceptable to use
$U_{a}\left(  t\right)  $ in place of $U_{d}\left(  t\right)  $ in calculating
the probability $\left\vert \left\langle n_{t}\right\vert U_{d}\left(
t\right)  \left\vert m_{0}\right\rangle \right\vert ^{2}$ when%
\begin{equation}
\left\vert \left\langle m_{t}\right\vert U_{d}\left(  t\right)  \left\vert
m_{0}\right\rangle \right\vert \approx1,\text{ }m=1,2. \label{b}%
\end{equation}
Eq. (\ref{b}) is the same as
\begin{equation}
\left\vert \left\langle m_{0}\right\vert U_{a}^{\dag}\left(  t\right)
U_{d}\left(  t\right)  \left\vert m_{0}\right\rangle \right\vert
\approx1,\text{ }m=1,2, \label{b1}%
\end{equation}
which we regard as the condition for the second kind of QAA. This kind of QAA
is considered in \cite{schiff,messiah}. When condition (\ref{b1})\ is
satisfied, the relative phase between the two instantaneous eigenstates in
$U_{d}\left(  t\right)  \left\vert \Psi\left(  0\right)  \right\rangle $ may
differ much from the relative phase in $U_{a}\left(  t\right)  \left\vert
\Psi\left(  0\right)  \right\rangle $.

The second kind of QAA is just a special case of the first kind and the
conditions for them are different: condition (\ref{a3}) can lead to (\ref{b1})
while (\ref{b1}) may not lead to (\ref{a3}), so there are cases where the
second kind of QAA is acceptable while the first kind of QAA is inacceptable.
Though the conditions for two kinds of QAA are given, usually it is not easy
to directly check whether they are satisfied. In the following we will give a
discussion on the conditions for them.

From (\ref{mm}) we know that if%
\begin{equation}
\left\vert c_{n}\left(  t\right)  -c_{n}\left(  0\right)  \right\vert
\ll1,\text{ }n=1,2,
\end{equation}
the first kind of QAA will be acceptable. From Eq. (\ref{ABC}) we can get%
\begin{equation}
\left\vert c_{1}\left(  t^{\prime\prime}\right)  -c_{1}\left(  0\right)
\right\vert \leq\left\vert \bar{A}\left(  t^{\prime\prime}\right)  \right\vert
+\left\vert \bar{B}\left(  t^{\prime\prime}\right)  \right\vert +\left\vert
\bar{C}\left(  t^{\prime\prime}\right)  \right\vert , \label{upp}%
\end{equation}
where%
\begin{equation}
\left\vert \bar{A}\left(  t^{\prime\prime}\right)  \right\vert =\left\vert
\frac{\hbar\left\langle 1_{t}\right\vert \frac{d}{dt}\left\vert 2_{t}%
\right\rangle }{1_{t}-2_{t}}\right\vert _{t=0}+\left\vert \frac{\hbar
\left\langle 1_{t}\right\vert \frac{d}{dt}\left\vert 2_{t}\right\rangle
}{1_{t}-2_{t}}\right\vert _{t=t^{\prime\prime}},
\end{equation}%
\begin{equation}
\left\vert \bar{B}\left(  t^{\prime\prime}\right)  \right\vert =t^{\prime
\prime}\max_{t\in\left[  0,t^{\prime\prime}\right]  }\left\vert \frac{d}%
{dt}\left(  \frac{\hbar\left\langle 1_{t}\right\vert \frac{d}{dt}\left\vert
2_{t}\right\rangle }{1_{t}-2_{t}}\right)  \right\vert , \label{fg}%
\end{equation}%
\begin{equation}
\left\vert \bar{C}\left(  t^{\prime\prime}\right)  \right\vert =t^{\prime
\prime}\max_{t\in\left[  0,t^{\prime\prime}\right]  }\left\vert \left\langle
2_{t}\right\vert \frac{d}{dt}\left\vert 1_{t}\right\rangle \frac
{\hbar\left\langle 1_{t}\right\vert \frac{d}{dt}\left\vert 2_{t}\right\rangle
}{1_{t}-2_{t}}\right\vert . \label{gf}%
\end{equation}
Eq. (\ref{upp}) gives an upper bound for $\left\vert c_{1}\left(
t^{\prime\prime}\right)  -c_{1}\left(  0\right)  \right\vert $. The first kind
of QAA will be acceptable if $\left\vert \bar{A}\left(  t^{\prime\prime
}\right)  \right\vert $, $\left\vert \bar{B}\left(  t^{\prime\prime}\right)
\right\vert $ and $\left\vert \bar{C}\left(  t^{\prime\prime}\right)
\right\vert $ are all very small. Eqs. (\ref{fg}) and (\ref{gf}) indicate that
it may be inappropriate to use $U_{a}\left(  t\right)  \left\vert \Psi\left(
0\right)  \right\rangle $ in place of $U_{d}\left(  t\right)  \left\vert
\Psi\left(  0\right)  \right\rangle $ when $t$ is large. When instantaneous
eigenstates and eigenvalues are given, it may be easy to check whether
$\left\vert \bar{A}\left(  t^{\prime\prime}\right)  \right\vert $, $\left\vert
\bar{B}\left(  t^{\prime\prime}\right)  \right\vert $ and $\left\vert \bar
{C}\left(  t^{\prime\prime}\right)  \right\vert $ are small at the given time
$t^{\prime\prime}$. For the example we will discuss in the next section, this
method is quite good.

The second kind of QAA is discussed in Refs. \cite{schiff,messiah}. It is
known that when%
\begin{equation}
\left\vert \int_{0}^{t}\left\langle m_{t_{1}}\right\vert \frac{d}{dt_{1}%
}\left\vert n_{t_{1}}\right\rangle e^{\frac{i}{\hbar}\int_{0}^{t_{1}}\left(
m_{t^{\prime}}-n_{t^{\prime}}\right)  dt^{\prime}}dt_{1}\right\vert \ll1,\quad
n\neq m, \label{xx}%
\end{equation}
the second kind of QAA will be acceptable \cite{messiah,schiff}. We emphasize
that the instantaneous eigenstates in (\ref{xx}) satisfy the phase condition
(\ref{r0}). We can derive Eq. (\ref{xx}) as follows. It is not hard to know
that the unitary operator $\tilde{U}\left(  t\right)  =U_{a}^{\dag}\left(
t\right)  U_{d}\left(  t\right)  $ is generated by the Hamiltonian
\begin{equation}
\tilde{H}_{t}=-i\hbar%
{\displaystyle\sum\limits_{m\neq n}^{M}}
\left\langle m_{t}\right\vert \frac{d}{dt}\left\vert n_{t}\right\rangle
e^{\frac{i}{\hbar}\int_{0}^{t}\left(  m_{t^{\prime}}-n_{t^{\prime}}\right)
dt^{\prime}}\left\vert m_{0}\right\rangle \left\langle n_{0}\right\vert .
\end{equation}
We can get an expansion of $\tilde{U}\left(  t\right)  $,
\begin{equation}
\tilde{U}\left(  t\right)  =I+\frac{1}{i\hslash}\int_{0}^{t}dt_{1}\tilde
{H}_{t_{1}}+\frac{1}{\left(  i\hslash\right)  ^{2}}\int_{0}^{t}dt_{1}\int
_{0}^{t_{1}}dt_{2}\tilde{H}_{t_{1}}\tilde{H}_{t_{2}}+\ldots. \label{add02}%
\end{equation}
Substitute Eq. (\ref{aeo}) and Eq. (\ref{add02}) into the expression
$U_{d}\left(  t\right)  =$ $U_{a}\left(  t\right)  \tilde{U}\left(  t\right)
$, we can obtain an expansion of $U_{d}\left(  t\right)  $. Similar
discussions appear in Refs. \cite{mosta,mac}. The condition (\ref{b1}) of the
second kind of QAA is equivalent to
\begin{equation}
\left\vert \left\langle m_{0}\right\vert U_{a}^{\dag}\left(  t\right)
U_{d}\left(  t\right)  \left\vert n_{0}\right\rangle \right\vert =\left\vert
\left\langle m_{0}\right\vert \tilde{U}\left(  t\right)  \left\vert
n_{0}\right\rangle \right\vert \ll1,\text{ }m\neq n. \label{add03}%
\end{equation}
When we just substitute the first two terms of Eq. (\ref{add02}) into
(\ref{add03}), we will obtain the condition (\ref{xx}). Though (\ref{xx}) as a
condition for the second kind of QAA has not been strictly proved since we
just use the first two terms in Eq. (\ref{add02}), we believe it is
sufficient. If $\left\langle m_{t_{1}}\right\vert \frac{d}{dt_{1}}\left\vert
n_{t_{1}}\right\rangle $ and $\left(  m_{t^{\prime}}-n_{t^{\prime}}\right)  $
are constants, condition (\ref{xx}) can be simplified into%
\begin{equation}
\left\vert \frac{\hbar\left\langle m_{t_{1}}\right\vert \frac{d}{dt_{1}%
}\left\vert n_{t_{1}}\right\rangle }{\left(  m_{t_{1}}-n_{t_{1}}\right)
}\right\vert \ll1,\text{ \ }m\neq n, \label{xx1}%
\end{equation}
while when $\left\langle m_{t_{1}}\right\vert \frac{d}{dt_{1}}\left\vert
n_{t_{1}}\right\rangle $ and $\left(  m_{t^{\prime}}-n_{t^{\prime}}\right)  $
are not constants, a modification of (\ref{xx1}) such as
\begin{equation}
\max_{t_{1}\in\left[  0,t\right]  }\left\vert \frac{\hbar\left\langle
m_{t_{1}}\right\vert \frac{d}{dt_{1}}\left\vert n_{t_{1}}\right\rangle
}{m_{t_{1}}-n_{t_{1}}}\right\vert \ll1,\text{ \ }m\neq n, \label{add05}%
\end{equation}
cannot replace (\ref{xx}) as a sufficient condition for the second of QAA
\cite{schiff,ms,tong,vertesi}, however it may be acceptable to regard
(\ref{add05}) as a sufficient condition when $\left\langle m_{t_{1}%
}\right\vert \frac{d}{dt_{1}}\left\vert n_{t_{1}}\right\rangle $ and $\left(
m_{t^{\prime}}-n_{t^{\prime}}\right)  $ vary slowly. The counterexample in
Ref. \cite{ms} demonstrates that a contradiction will appear at a special
evolution time if we take (\ref{add05}) as a sufficient condition. We think
the rapid changing of $\left\langle m_{t_{1}}\right\vert \frac{d}{dt_{1}%
}\left\vert n_{t_{1}}\right\rangle $ causes the contradiction because when
$\left\langle m_{t_{1}}\right\vert \frac{d}{dt_{1}}\left\vert n_{t_{1}%
}\right\rangle $ varies fast, especially when its varying frequency resonant
to the energy gap, condition (\ref{add05}) deviates much from (\ref{xx}).
Finally we want to emphasize that even when condition (\ref{xx}) is satisfied,
it doesnot means it is acceptable to use $U_{a}\left(  t\right)  \left\vert
n_{0}\right\rangle $, condition \ to approximate $U_{d}\left(  t\right)
\left\vert n_{0}\right\rangle $.

\section{An example\label{sec:4}}

In this section we will give an example to demonstrate that there are cases
where the second kind of QAA is acceptable while the first kind is inacceptable.

Consider the Hamiltonian%
\begin{equation}
H_{t}=-\hbar\omega_{0}\left[  \sigma_{x}\cos2\omega t+\sigma_{y}\sin2\omega
t\right]  , \label{H}%
\end{equation}
where $\omega_{0}$ and $\omega$ are both positive. Its eigenvalues are
$1_{t}=\hbar\omega_{0}$ and $2_{t}=-\hbar\omega_{0}$. We choose the
eigenstates
\begin{align}
\left\vert 1_{t}\right\rangle  &  =\left(  e^{-i\omega t}\left\vert
0\right\rangle -e^{i\omega t}\left\vert 1\right\rangle \right)  /\sqrt{2},\\
\left\vert 2_{t}\right\rangle  &  =\left(  e^{-i\omega t}\left\vert
0\right\rangle +e^{i\omega t}\left\vert 1\right\rangle \right)  /\sqrt
{2},\nonumber
\end{align}
which satisfy the phase condition (\ref{r0}). It can be proved that in the
basis $\left\{  \left\vert 1_{0}\right\rangle ,\left\vert 2_{0}\right\rangle
\right\}  $
\begin{equation}
U_{a}^{\dag}\left(  t\right)  U_{d}\left(  t\right)  =\left[
\begin{array}
[c]{cc}%
e^{i\omega_{0}t}\left(  \cos\bar{\omega}t-\frac{i\omega_{0}}{\bar{\omega}}%
\sin\bar{\omega}t\right)  & e^{i\omega_{0}t}\times\frac{i\omega}{\bar{\omega}%
}\sin\bar{\omega}t\\
e^{-i\omega_{0}t}\times\frac{i\omega}{\bar{\omega}}\sin\bar{\omega}t &
e^{-i\omega_{0}t}\left(  \cos\bar{\omega}t+\frac{i\omega_{0}}{\bar{\omega}%
}\sin\bar{\omega}t\right)
\end{array}
\right]  ,
\end{equation}
where $\bar{\omega}=\sqrt{\omega^{2}+\omega_{0}^{2}}$.

The condition (\ref{b1}) for the second kind of QAA can be written as%
\begin{equation}
\left(  \cos\bar{\omega}t\right)  ^{2}+\left(  \frac{\omega_{0}}{\bar{\omega}%
}\sin\bar{\omega}t\right)  ^{2}\approx1, \label{dr1}%
\end{equation}
which can be simplified into%
\begin{equation}
\left(  \left(  \frac{\omega_{0}}{\bar{\omega}}\right)  ^{2}-1\right)
\sin^{2}\bar{\omega}t\approx0. \label{dr2}%
\end{equation}
From (\ref{dr2}) it can be seen that when $\omega\ll\omega_{0}$ the second
kind of QAA will be acceptable no matter how lage $t$ is. This result can also
be obtained from Eq. (\ref{xx}).

The condition (\ref{a3}) for the first kind of QAA can be written as%
\begin{equation}
e^{i\omega_{0}t}\left(  \cos\bar{\omega}t-\frac{i\omega_{0}}{\bar{\omega}}%
\sin\bar{\omega}t\right)  \approx1, \label{dr3}%
\end{equation}
which is the same as
\begin{equation}
e^{-i\frac{\omega^{2}t}{\bar{\omega}+\omega_{0}}}+ie^{i\omega_{0}t}\left(
1-\frac{\omega_{0}}{\bar{\omega}}\right)  \sin\bar{\omega}t\approx1.
\label{dr4}%
\end{equation}
Notice that (\ref{dr3}) can lead to (\ref{dr1}), it is not hard to see that
(\ref{dr4}) is equivalent to (\ref{dr2}) plus $\omega^{2}t/\left(  \bar
{\omega}+\omega_{0}\right)  \ll1$. When $\omega\ll\omega_{0}$ and $\omega
^{2}t/\omega_{0}\ll1$ the first kind of QAA will be acceptable. Though
condition $\omega\ll\omega_{0}$ is enough for the second kind of QAA to be
acceptable, it cannot certainly\ make the first kind of QAA acceptable when
$t$ is large and this result has been implied in Ref. \cite{pla}.

In this example, the upper bound expression Eq. (\ref{upp}) can be written as%
\begin{equation}
\left\vert c_{1}\left(  t\right)  -c_{1}\left(  0\right)  \right\vert
\leq\frac{\omega}{\omega_{0}}+\frac{\omega^{2}t}{2\omega_{0}}. \label{dr6}%
\end{equation}
From (\ref{dr6}) we can also conclude that when $\omega\ll\omega_{0}$ and
$\omega^{2}t/\omega_{0}\ll1$ the first kind of QAA will be acceptable, which
coincides with the result obtained by figuring out $U_{a}^{\dag}\left(
t\right)  U_{d}\left(  t\right)  $.

\section{Quantum adiabatic approximation for Hamiltonian $H\left(  t/T\right)
$\label{sec:5}}

The above discussion on QAA applies to a general time-dependent $H_{t}$. When
$H_{t}$ has the special form $H_{t}=H\left(  s\right)  $ where $s=t/T\in
\left[  0,1\right]  $, we can discuss the relation between the parameter $T$
and the error caused by QAA.

\begin{theorem}
\label{th2}Suppose $H_{t}$ has the special form $H_{t}=H\left(  s\right)  $
where $s=t/T\in\left[  0,1\right]  $. When
\begin{equation}
T\geq\frac{\sqrt{2}\hbar}{\delta}\max_{s\in\left[  0,1\right]  }\left(
\frac{2\left\Vert \frac{dH\left(  s\right)  }{ds}\right\Vert }{\Delta\left(
s\right)  ^{2}}+\frac{7\left\Vert \frac{dH\left(  s\right)  }{ds}\right\Vert
^{2}}{\Delta\left(  s\right)  ^{3}}+\frac{\left\Vert \frac{d^{2}H\left(
s\right)  }{ds^{2}}\right\Vert }{\Delta\left(  s\right)  ^{2}}\right)
\label{add04}%
\end{equation}
where $\Delta\left(  s\right)  =\left\vert 1\left(  s\right)  -2\left(
s\right)  \right\vert $, the error caused by QAA cannot be bigger than
$\delta$, i.e.,
\begin{equation}
\left\Vert U_{d}\left(  t\right)  \left\vert \Psi\left(  0\right)
\right\rangle -U_{a}\left(  t\right)  \left\vert \Psi\left(  0\right)
\right\rangle \right\Vert \leq\delta.
\end{equation}

\end{theorem}

\begin{proof}
In Sec. \ref{sec:2}, we write $U_{d}\left(  t\right)  \left\vert \Psi\left(
0\right)  \right\rangle $ in Eq. (\ref{nn}) and we know from Eqs.
(\ref{ABC},\ref{a}-\ref{c}) that
\begin{align}
&  \left\vert c_{n}\left(  t\right)  -c_{n}\left(  0\right)  \right\vert
\label{x2}\\
&  \leq\frac{1}{T}\max_{s\in\left[  0,1\right]  }\left[  2\left\vert f\left(
s\right)  \right\vert +\left\vert \frac{df\left(  s\right)  }{ds}\right\vert
+\left\vert \left\langle 2\left(  s\right)  \right\vert \frac{d}{ds}\left\vert
1\left(  s\right)  \right\rangle f\left(  s\right)  \right\vert \right]
,\nonumber
\end{align}
where
\begin{equation}
f\left(  s\right)  =\frac{\hbar\left\langle 1\left(  s\right)  \right\vert
\frac{d}{ds}\left\vert 2\left(  s\right)  \right\rangle }{i\left(  1\left(
s\right)  -2\left(  s\right)  \right)  }%
\end{equation}
is defined in Eq. (\ref{f}). Now we give an upper bound of $\left\vert
c_{n}\left(  t\right)  -c_{n}\left(  0\right)  \right\vert $ expressed in norm
of the Hamiltonian $H\left(  s\right)  $ and its derivatives. First we have
\cite{messiah,schiff}
\begin{equation}
\left\langle 1\left(  s\right)  \right\vert \frac{d}{ds}\left\vert 2\left(
s\right)  \right\rangle =\frac{\left\langle 1\left(  s\right)  \right\vert
\frac{dH\left(  s\right)  }{ds}\left\vert 2\left(  s\right)  \right\rangle
}{2\left(  s\right)  -1\left(  s\right)  },
\end{equation}
so there are
\begin{equation}
\left\vert f\left(  s\right)  \right\vert =\left\vert \frac{\hbar\left\langle
1\left(  s\right)  \right\vert \frac{dH\left(  s\right)  }{ds}\left\vert
2\left(  s\right)  \right\rangle }{-i\left(  1\left(  s\right)  -2\left(
s\right)  \right)  ^{2}}\right\vert \leq\frac{\hbar\left\Vert \frac{dH\left(
s\right)  }{ds}\right\Vert }{\Delta\left(  s\right)  ^{2}}, \label{qq1}%
\end{equation}
and%
\begin{equation}
\left\vert \left\langle 2\left(  s\right)  \right\vert \frac{d}{ds}\left\vert
1\left(  s\right)  \right\rangle f\left(  s\right)  \right\vert \leq
\frac{\hbar\left\Vert \frac{dH\left(  s\right)  }{ds}\right\Vert ^{2}}%
{\Delta\left(  s\right)  ^{3}}. \label{qq2}%
\end{equation}
It can be proved that%
\begin{align}
\left\vert \frac{d}{ds}f\left(  s\right)  \right\vert  &  =\left\vert \frac
{d}{ds}\left[  \frac{\hbar\left\langle 1\left(  s\right)  \right\vert
\frac{dH\left(  s\right)  }{ds}\left\vert 2\left(  s\right)  \right\rangle
}{-i\left(  1\left(  s\right)  -2\left(  s\right)  \right)  ^{2}}\right]
\right\vert \label{qq3}\\
&  \leq\frac{\hbar\left\Vert \frac{d^{2}H\left(  s\right)  }{ds^{2}%
}\right\Vert }{\Delta\left(  s\right)  ^{2}}+6\frac{\hbar\left\Vert
\frac{dH\left(  s\right)  }{ds}\right\Vert ^{2}}{\Delta\left(  s\right)  ^{3}%
}.\nonumber
\end{align}
Substitute Eqs. (\ref{qq1}-\ref{qq3}) into Eq. (\ref{x2}), we get%
\begin{align}
&  \left\vert c_{n}\left(  t\right)  -c_{n}\left(  0\right)  \right\vert \\
&  \leq\frac{\hbar}{T}\max_{s\in\left[  0,1\right]  }\left[  \frac{2\left\Vert
\frac{dH\left(  s\right)  }{ds}\right\Vert }{\Delta\left(  s\right)  ^{2}%
}+\frac{\left\Vert \frac{d^{2}H\left(  s\right)  }{ds^{2}}\right\Vert }%
{\Delta\left(  s\right)  ^{2}}+\frac{7\left\Vert \frac{dH\left(  s\right)
}{ds}\right\Vert ^{2}}{\Delta\left(  s\right)  ^{3}}\right]  .\nonumber
\end{align}
Because%
\begin{align}
&  \left\Vert U_{d}\left(  t\right)  \left\vert \Psi\left(  0\right)
\right\rangle -U_{a}\left(  t\right)  \left\vert \Psi\left(  0\right)
\right\rangle \right\Vert \nonumber\\
&  =\sqrt{%
{\displaystyle\sum\limits_{n=1}^{2}}
\left\vert c_{n}\left(  t\right)  -c_{n}\left(  0\right)  \right\vert ^{2}},
\end{align}
when (\ref{add04}) is satisfied we have
\begin{align}
&  \left\Vert U_{d}\left(  t\right)  \left\vert \Psi\left(  0\right)
\right\rangle -U_{a}\left(  t\right)  \left\vert \Psi\left(  0\right)
\right\rangle \right\Vert \\
&  \leq\frac{\hbar\sqrt{2}}{T}\max_{s\in\left[  0,1\right]  }\left[
\frac{2\left\Vert \frac{dH\left(  s\right)  }{ds}\right\Vert }{\Delta\left(
s\right)  ^{2}}+\frac{\left\Vert \frac{d^{2}H\left(  s\right)  }{ds^{2}%
}\right\Vert }{\Delta\left(  s\right)  ^{2}}+\frac{7\left\Vert \frac{dH\left(
s\right)  }{ds}\right\Vert ^{2}}{\Delta\left(  s\right)  ^{3}}\right]
\nonumber\\
&  \leq\delta,
\end{align}
which completes the proof.

A result similar to theorem \ref{th2} appears in Ref. \cite{regev}, which is
derived from a very different method. Under the same error $\delta$, the
required evolution time $T$ in our result is proportional to $1/\delta$ while
in Ref. \cite{regev} it is proportional to $1/\delta^{2}$. We think this is
one of the most differences between the two results.
\end{proof}

\section{Conclusion\label{sec:6}}

We give a discussion on the conditions of QAA. We present a proof of QAT,
which we think is easier than that appears in the textbook \cite{messiah}. We
think there are two kinds of QAA, one cares the relative phase in the
approximate system state while the other does not. For the kind of Hamiltonian
$H\left(  t/T\right)  $ we give a relation between the size of the error
caused by QAA and the parameter $T$.

\begin{acknowledgments}
This work was funded by the National Fundamental Research Program, the
National Natural Science Foundation of China (Grant Nos. 10674127, 60121503),
the Innovation funds from Chinese Academy of Science and program for NCET.
\end{acknowledgments}

\end{document}